\documentclass[a4paper,envcountsect,envcountsame]{llncs}
\newcommand{\composedUprime}{%
  \mathrel{\vbox{\offinterlineskip\ialign{%
    \hfil##\hfil\cr
    $\scriptscriptstyle\circ$\cr
    \noalign{\kern0.1ex}
    $\boldU'$\cr
}}}}

\newcommand{\composedU}{%
  \mathrel{\vbox{\offinterlineskip\ialign{%
    \hfil##\hfil\cr
    $\scriptscriptstyle\circ$\cr
    \noalign{\kern0.1ex}
    $\boldU$\cr
}}}}

\newcommand{\composedW}{%
  \mathrel{\vbox{\offinterlineskip\ialign{%
    \hfil##\hfil\cr
    $\scriptscriptstyle\circ$\cr
    \noalign{\kern0.1ex}
    $\boldW$\cr
}}}}

\usepackage{times}

\usepackage{microtype, amsmath,amssymb, amsfonts, amsbsy, pstricks,bbm,mathrsfs, multirow, bigstrut,graphicx,fixmath, enumerate}
 
\usepackage{dingbat}
\usepackage[english]{babel} 
\usepackage{proof}

\newenvironment{frcseries}{\fontfamily{frc}\selectfont}{}

\newcommand{\path}{\mathfrak{p}}
\newcommand{\apartialorder}{\ell}
\newcommand{\partialorders}{\mathcal{P}}
\newcommand{\alphabettransitions}{{T}}
\newcommand{\process}{{\pi}}
\newcommand{\directedslicealphabet}{\overrightarrow{\mathbold{\Sigma}}}

\newcommand{\width}{w}
\newcommand{\labelAlphabet}{{\alphabettransitions}}

\newcommand{\unitdecompositions}{\mathbold{ud}}
\newcommand{\transitiveClosure}{\mathbold{tc}}
\newcommand{\transitiveReduction}{\mathbold{tr}}
\newcommand{\msographs}{\mathit{MSO}_2}
\newcommand{\msopartialorders}{\mathit{MSO}}

\newcommand{\automaton}{{\mathcal{A}}}

\newcommand{\boldS}{{\mathbf{S}}}
\newcommand{\boldU}{{\mathbf{U}}}
\newcommand{\lang}{{\mathcal{L}}}

\newcommand{\N}{{\mathbb{N}}}

  \newcommand{\takesp}{{\hat{p}}} \newcommand{\putsp}{{\check{p}}}

\newcommand{\keywords}[1]{\par\addvspace\baselineskip\noindent\keywordname\enspace\ignorespaces#1}

\newcommand{\petriNet}{{N}}
\newcommand{\graph}{{\mathcal{G}}}
\newcommand{\partialOrder}{{\mathit{po}}}

\begin{document}

\title{Automated Verification, Synthesis and Correction of Concurrent Systems via MSO Logic   
\vspace{-0.5cm}}
\author{
Mateus de Oliveira Oliveira}
\institute{KTH Royal Institute of Technology \\  mdeoliv@kth.se}

\maketitle

\vspace{-0.7cm}
\begin{abstract} \noindent 
In this work we provide algorithmic solutions to five fundamental problems concerning the 
verification, synthesis and correction of concurrent systems that can be modeled by  
bounded $p/t$-nets. We express concurrency via partial orders and assume that behavioral specifications
are given via monadic second order logic.
A $c$-partial-order is a partial order 
whose Hasse diagram can be covered by $c$ paths. 
For a finite set $T$ of transitions, we let $\partialorders(c,T,\varphi)$ denote the set of all $T$-labelled $c$-partial-orders satisfying $\varphi$. 
If $N=(P,T)$ is a $p/t$-net we let $\partialorders(N,c)$ denote the set of 
all $c$-partially-ordered runs of $N$. A $(b,r)$-bounded $p/t$-net is a $b$-bounded $p/t$-net
in which each place appears repeated at most $r$ times. We solve the following problems:  

\begin{enumerate}
	\item \textbf{Verification:} 
		given an MSO formula $\varphi$ and a bounded $p/t$-net $N$ determine 
		whether $\partialorders(N,c)\subseteq \partialorders(c,T,\varphi)$, whether 
		$\partialorders(c,T,\varphi)\subseteq \partialorders(N,c)$, or whether 
		$\partialorders(N,c)\cap \partialorders(c,T,\varphi) = \emptyset$. 
	\item \textbf{Synthesis from MSO Specifications:} 
		given an MSO formula $\varphi$, synthesize a semantically minimal $(b,r)$-bounded $p/t$-net
		$N$ satisfying $\partialorders(c,T,\varphi)\subseteq \partialorders(N,c)$. 	
	\item \textbf{Semantically Safest Subsystem: }  given an MSO formula $\varphi$ defining a set of safe partial orders, and a $b$-bounded $p/t$-net $N$, 
		possibly containing unsafe behaviors, synthesize the safest $(b,r)$-bounded $p/t$-net $N'$ whose behavior lies in between 
		$\partialorders(N,c)\cap \partialorders(c,T,\varphi)$ and $\partialorders(N,c)$. 
	\item \textbf{Behavioral Repair: } 
		 given two MSO formulas $\varphi$ and $\psi$, and a $b$-bounded $p/t$-net $N$, 
		synthesize a semantically minimal $(b,r)$-bounded $p/t$ net $N'$ whose behavior lies in between 
		$\partialorders(N,c)\cap \partialorders(c,T,\varphi)$ and $\partialorders(c,T,\psi)$. 
	\item \textbf{Synthesis from Contracts:}
		given an MSO formula $\varphi^{\mathit{yes}}$ specifying a set of good behaviors and an MSO formula $\varphi^{\mathit{no}}$ specifying 
		a set of bad behaviors, synthesize a semantically minimal $(b,r)$-bounded $p/t$-net $N$ 
		such that $\partialorders(c,T,\varphi^{\mathit{yes}}) \subseteq \partialorders(N,c)$ but $\partialorders(c,T,\varphi^{\mathit{no}})\cap \partialorders(N,c) = \emptyset$. 
\end{enumerate}

 \end{abstract}

\vspace{-1cm}
\section{Introduction}
\label{section:Introduction}
\vspace{-0.2cm}
{\em Model checking} and {\em system synthesis} are two complementary paradigms that 
are widely used to provide correctness guarantees for computational systems. On the one hand, the goal of model 
checking is to verify whether the behavior of a given system is in accordance with a given specification 
\cite{ClarkeEmerson1986,ClarkeGrumbergPeled1999,Pnueli1981,QueilleSifakis1982}. On the other hand, the goal 
of system synthesis is to mechanically construct a system from a behavioral specification \cite{CernyChatterjeeHenzingerRadhakrishnaSingh2011,ClarkeEmerson1982,KupfermanVardiYannakakis2011,MannaWolper1984,PnueliRosner1989B}. 
When combined, model checking and synthesis can be used as primitives for the development of 
powerful methodologies aimed at the mechanical correction of bugs, such as system repair
\cite{GriesmayerBloemCook2006,JobstmannGriesmayerBloem2005,SamantaDeshmokhEmerson2008,EssenJobstmann2013}.
In this work we develop a combined theory of {\em model checking} and {\em system synthesis} that 
is fully compatible with the partial order theory of concurrency. Our systems are modeled via 
bounded place/transition nets, while our behavioral specifications are 
given in monadic second order logic. We solve five fundamental problems lying in the intersection of 
system verification, system synthesis and system repair. First we show how to compare the partial order 
behavior of bounded $p/t$-nets with  partial order behaviors specified via MSO formulas. Second, we show 
how to synthesize bounded $p/t$-nets from MSO-definable sets of partial orders. Third, we show how 
to obtain the {\em semantically safest subsystem} of a bounded $p/t$-net with respect to a MSO specification. 
Fourth, we transpose the methodology of program repair introduced by Jobstmann and von Essen \cite{EssenJobstmann2013}
to the context of bounded $p/t$-nets with partial-order runs. Finally, we show how to synthesize bounded
$p/t$-nets from partial-order contracts. 
Before giving a precise definition of each of the problems described above, we briefly introduce the main elements 
of our model. 

\vspace{-0.2cm}
\paragraph{\textbf{Bounded place/transition nets: }}
Petri nets  \cite{Petri1962A}, also known as place/transition-nets are recognized as an elegant mathematical formalism for the
specification of concurrent systems. During the last four decades, $p/t$-nets have found applications in the modeling of 
real time fault tolerant systems, faulty critical systems, communication protocols, 
logic controllers, and many others types of computing systems \cite{Murata1989,ZurawskiZhou1994}.
A $p/t$-net consists of a multiset of places, which are initially loaded with a set of tokens, 
and a set of transitions. The disposition of the tokens among the places of a $p/t$-net determine which transitions are allowed 
to fire. A transition, by its turn, when fired, removes tokens from some places and adds tokens to some places.
In this work we will be concerned with the partial order theory of bounded $p/t$-nets. We say that 
a $p/t$-net is $b$-bounded if after firing any sequence of transitions, the number of tokens in 
each of its places remains bounded by $b$, and that a $p/t$-net is $(b,r)$-bounded if it is $b$-bounded and if each place 
occurs in it at most $r$ times. We will define $p/t$-nets more precisely in Section \ref{section:PetriNets}.

\vspace{-0.3cm}
\paragraph{\textbf{Partial Orders:}}

When concurrency is interpreted accordingly to the interleaving semantics, the
execution of concurrent actions is identified with the non-deterministic choice among
all possible orders in which such actions can occur. Although satisfactory for many
applications of practical relevance, this point of view has some drawbacks. First,
the interleaving semantics is not compatible with the notion of action refinement,
in which an atomic action is replaced by a set of sub-actions \cite{GlabbeekGoltz2001}. 
Second, this point of view is not appropriate to model concurrent scenarios in which several users
have concurrent read/write access to databases \cite{FleRoucairol1982} nor to model the behavior 
of read/write operations in multiprocessors that implement {\em weak memory models} \cite{AlglaveKroeningTautschnig2013}. 

A well established point of view which is able to overcome these drawbacks is to represent both concurrency 
and causality as partially ordered sets of events \cite{Gischer1988,GaifmanPratt1987,Vogler1992}. 
The partial order imposed on a set of events can be interpreted according to two standard semantics. 
According to the first, the {\em causal semantics}, an event $v$ is smaller than an event $v'$ if 
$v'$ causally depends on the occurrence of $v$. According to the second, the execution semantics, the fact that 
$v$ is smaller than $v'$ simply indicates that $v$ does not occur after $v'$. However, in this case the event $v$ may not 
necessarily be one of the causes of the event $v'$. In this work, we will study the partial order behavior of 
bounded $p/t$-nets according to both semantics. 
 
\paragraph{\textbf{$c$-Partial-Orders}:}
We introduce a new parameterization for the study of the partial order behavior of concurrent 
systems. Recall that the Hasse diagram of a partial order $\ell$ is the directed acyclic graph $H$ with the least number 
of edges whose transitive closure equals $\apartialorder$. We say that a partial order $\ell$ is a 
$c$-partial-order if its Hasse diagram $H$ can be covered by $c$ paths. In other words if there exist 
paths $\path_1,...,\path_k$ in $H$ such that $H=\cup_{i=1}^c \path_i$. We notice that the paths are not 
assumed to be edge disjoint nor vertex disjoint. We let $\partialorders_{cau}(N,c)$ denote the set 
of all $c$-partial-orders that can be associated to a $p/t$-net $N$ according to the causal semantics, and by $\partialorders_{ex}(N,c)$
the set of all $c$-partial-orders that can be associated with $N$ according to the execution semantics.
Intuitively, the parameter $c$ characterizes the thickness of the partial order, and provides a width measure 
that is stronger and more algorithmically friendly than than the traditional notion of width used in partial order theory.
We observe that the execution behavior $\partialorders_{ex}(N,1)$ is simply the set of all possible firing sequences of $N$. 
If $N$ is a $b$-bounded $p/t$-net with $n$ places then the set $\partialorders_{cau}(N,b\cdot n)$ already 
comprises all possible causal runs of $N$. We contrast this observation with the fact that there are very 
simple examples\footnote{For instance a net consisting of two places $p_1,p_2$, initialize with a unique token each, and 
two transitions $t_1,t_2$ such that $t_i$ takes one token from $p_i$ and puts it back on $p_i$.} of $p/t$-nets whose execution behavior $\partialorders_{ex}(N,c)$ is strictly contained into 
$\partialorders_{ex}(N,c+1)$ for each $c\in \N$.%

\paragraph{\textbf{Monadic Second Order Logic of Graphs: }}
The monadic second order logic over partial orders extends first order logic 
by adding the possibility of quantifying over sets of vertices. The role of 
MSO logic in the study of the partial order behavior of concurrent
systems was emphasized in \cite{Madhusudan2001} in the context of the theory 
of message sequence graphs. Let $T$ be a finite set of symbols, which should be regarded as labels of transitions in a concurrent system. 
We say that a partial order $\apartialorder$ is a $T$-labeled partial order if each of its nodes 
are labeled with some element of $T$. Let $\varphi$ be an $\msopartialorders$ formula expressing a property of $T$-labeled 
partial orders, and let $c\in N$. We denote by $\partialorders(c,T,\varphi)$ the set of all $T$-labeled $c$-partial-orders 
satisfying $\varphi$. In this work the connection between the $c$-partial-order 
behavior of bounded $p/t$-nets and MSO logic will be established via a formalism 
called slice automaton (Section \ref{section:RegularSliceLanguages}). In the context 
of this paper, slice automata should be regarded as a generalization of message sequence 
graphs which is suitable for the representation of the $c$-partial-order behavior of bounded $p/t$-nets. 
Indeed, we will show that for each MSO formula $\varphi$, the set of all $c$-partial-orders 
satisfying $\varphi$ can be represented by a slice automaton. The connection 
with bounded $p/t$-nets stems from a result previously proved by us \cite{deOliveiraOliveira2013}
stating that the $c$-partial-order behavior of bounded $p/t$-nets can also be 
effectively represented via slice automata.

In the next five subsections we will state our main results and establish further 
connections with existing literature.

\vspace{-10pt}
\subsection{Verification of the Partial Order Behavior of Bounded $p/t$-Nets.}
\label{subsection:Verification}

Suppose we have a concurrent system modeled by a $b$-bounded $p/t$-net $N$ and let 
$\varphi$ be an MSO formula. In Theorem \ref{theorem:MSOAndPetriNets} below we address 
three verification results. First, assuming that $\varphi$ defines a set of faulty behaviors, 
we can mechanically determine whether or not some of the partial order runs of $N$ is faulty (Theorem \ref{theorem:MSOAndPetriNets}.\ref{MSOAndPetriNets-i}).
Second, on the contrapositive, assuming that $\varphi$ specifies a set of good partial order behaviors, we can 
test whether or not all behaviors of $N$ are good (Theorem \ref{theorem:MSOAndPetriNets}.\ref{MSOAndPetriNets-ii}). Third, 
assuming that $\varphi$ specify a set of desired partial order behaviors, we can decide whether the 
partial order behavior of $N$ comprises all partial orders specified by $\varphi$ (Theorem \ref{theorem:MSOAndPetriNets}.\ref{MSOAndPetriNets-iii}). 
All three verification results hold with respect with both the execution and the causal semantics. 
Below, we let the variable $sem$ be equal to $ex$ if we are considering the execution semantics and equal to $cau$ 
if we are considering the causal semantics. A precise definition of how partial orders are assigned to 
$p/t$-nets according to each of these semantics will be given in Section \ref{section:PetriNets}.  
\vspace{-5pt}
\begin{theorem}[Verification]
\label{theorem:MSOAndPetriNets} 
Let $\varphi$ be an $\msopartialorders$ formula, $\petriNet$ be a $b$-bounded $p/t$-net, $c\in \N$, and $sem \in \{ex,cau\}$.
\vspace{-5pt} 
\begin{enumerate}[i)]
	\item \label{MSOAndPetriNets-i} One may effectively determine whether $\partialorders_{sem}(N,c) \cap \partialorders(c,T,\varphi) = \emptyset$. 
	\item \label{MSOAndPetriNets-ii} One may effectively determine whether $\partialorders_{sem}(N,c) \subseteq \partialorders(c,T,\varphi)$. 
	\item \label{MSOAndPetriNets-iii} One may effectively determine whether $\partialorders(c,T,\varphi) \subseteq \partialorders_{sem}(N,c)$. 
\end{enumerate}
\end{theorem}
%\vspace{-5pt}
Notice that Theorems  \ref{theorem:MSOAndPetriNets}.\ref{MSOAndPetriNets-i} and \ref{theorem:MSOAndPetriNets}.\ref{MSOAndPetriNets-ii} 
can be reduced to each other 
since $\partialorders_{sem}(N,c) \cap \partialorders(c,T,\varphi) = \emptyset$ if and only if $\partialorders_{sem}(N,c) \subseteq  \partialorders(c,T,\neg\varphi)$.
Theorem \ref{theorem:MSOAndPetriNets} addresses for the first time safety and conformance tests of the both the execution 
and the  causal behaviors of general bounded $p/t$-nets. We notice that in the special case of {\em pure}\footnote{ A $p/t$-net is pure if no transition takes a token from a place and 
puts it in the same place.} bounded $p/t$-nets, the model checking of the causal behavior was addressed in \cite{AvellanedaMorin2013} using the machinery of vector addition systems with states.
In our notation, this corresponds to testing whether $\partialorders_{cau}(N,n\cdot b) \cap \partialorders(n\cdot b,T,\varphi) =\emptyset$, where $n$ is the number of places of 
$N$ and $b$ its bound. However, as pointed out in \cite{BadouelDarondeau1996a}, pure $p/t$-nets are a rather restricted subclass of $p/t$-nets, since they are 
not able to model for instance, waiting loops in communication protocols. The results in \cite{AvellanedaMorin2013} are not able to address the model checking of $p/t$-nets according to the 
execution semantics, nor to provide an analog of Theorem \ref{theorem:MSOAndPetriNets}.\ref{MSOAndPetriNets-iii} with respect to neither the causal nor the execution semantics. 

\vspace{-10pt}
\subsection{Synthesis of Bounded $p/t$-Nets from MSO Specifications}
\label{subsection:Synthesis}
\vspace{-5pt}
In our second result (Theorem \ref{theorem:MSOSynthesis}) we address the synthesis of $p/t$-nets from $\msopartialorders$ 
definable sets of $c$-partial-orders. We say that a $p/t$-net is $(b,r)$-bounded if each place occurs with multiplicity at most $r$ 
and if each place has at most $b$-tokens on each legal marking of $N$. 
Let $\mathit{sem}\in \{\mathit{ex},\mathit{cau}\}$. 
A $(b,r)$-bounded $p/t$-net $\petriNet$ is $c$-$\mathit{sem}$-minimal for a partial order 
language $\partialorders$, if $\partialorders \subseteq \partialorders_{\mathit{sem}}(\petriNet,c)$ and if there is no other $(b,r)$-bounded $p/t$-net 
$\petriNet'$ with $\partialorders \subseteq \partialorders_{\mathit{sem}}(\petriNet',c)\subsetneq \partialorders_{\mathit{sem}}(\petriNet,c)$. 
\begin{theorem}[Synthesis]
\label{theorem:MSOSynthesis} 
Let $\varphi$ be an $\msopartialorders$ formula, $T$ be a finite set of transitions, and $b,c,r \in\N$. 
Then one can effectively determine whether there exists a $(b,r)$-bounded $p/t$-net $\petriNet$ which is $c$-$\mathit{sem}$-minimal for 
$\partialorders(c,T,\varphi)$. In the case such a net $\petriNet$ exists one can effectively construct 
it. 
\end{theorem}
Theorem \ref{theorem:MSOSynthesis} addresses for the first time the synthesis of bounded $p/t$-nets from MSO 
specifications. Observe that the minimality condition imposed in theorem \ref{theorem:MSOSynthesis} implies 
that in the case that it is not possible to synthesize a net precisely matching the specification $\varphi$, 
it is still possible to synthesize a net with the fewest number of bad partial-order runs as possible. 
We observe that the parameter $r$ in Theorem \ref{theorem:MSOSynthesis} is only relevant when considering the 
causal semantics. Indeed, adding repeated places to $p/t$-nets does not change their execution behavior. However 
the addition of repeated places can indeed increase the causal behavior of $p/t$-net \cite{deOliveiraOliveira2010}.
We also observe that when considering the synthesis with the execution semantics all $c$-ex-minimal $p/t$-nets
for a partial order language $\partialorders$ have the same partial order behavior. However when considering 
the causal semantics, there may exist two $p/t$-nets $N_1$ and $N_2$ whose behavior is $c$-cau-minimal for $\partialorders$, but 
for which $\partialorders(N_1,c)\neq \partialorders(N_2,c)$. 

It is worth comparing Theorem \ref{theorem:MSOSynthesis} with existing literature. When considering the interleaving 
semantics of bounded $p/t$-nets, the synthesis problem from regular languages was studied extensively in \cite{BadouelDarondeau1996a,BadouelDarondeau1998,Darondeau1998,Darondeau2000} 
via a set of combinatorial techniques called {\em theory of regions}. Thus, via Buchi-Elgot Theorem stating that MSO Logic over strings 
is as expressive as regular languages \cite{Buchi1960}, the theory of regions 
can be used to synthesize nets whose 
interleaving behavior satisfies a given MSO formula over {\em strings}. In our notation this corresponds
to synthesizing a bounded net $N$ whose $1$-execution-behavior $\partialorders_{ex}(N,1)$ is $1$-ex-minimal with respect to 
$\partialorders(\varphi,T,1)$. Here we solve the synthesis problem from MSO languages for any $c\geq 1$, and with respect 
to both the causal and execution semantics. It is worth noting that the synthesis of bounded $p/t$-nets 
with the execution semantics from certain restricted partial order formalisms that are not able to represent the 
behavior of bounded $p/t$-nets was considered in \cite{BergenthumDeselLorenzMauserFundamenta2008,BergenthumDeselLorenzMauserInfinite2008}, but 
no connection with logic was established therein. The synthesis of bounded $p/t$-nets from a mathematical object 
that is able to fully represent the causal behavior of bounded $p/t$-nets was solved by us in \cite{deOliveiraOliveira2010}, 
solving in this way an open problem stated in \cite{JuhasLorenzDesel2005}. This mathematical object is called slice automaton, 
and will be described in Section \ref{section:RegularSliceLanguages}. The proof of Theorem \ref{theorem:MSOSynthesis} follows
by establishing a non-trivial connection between monadic second order logic and slice automata. 

\vspace{-10pt}
\subsection{Semantically Safest Subsystem}
\label{subsection:SemanticallySafestSubsystem}
\vspace{-5pt}
Suppose that we have in hands an MSO formula $\varphi$ specifying a set of safe behaviors, and 
a concurrent system specified by a $(b,r)$-bounded $p/t$-net $\petriNet$. 
Suppose that after verifying $N$ according to Theorem \ref{theorem:MSOAndPetriNets} we discover that some 
runs of $\petriNet$ are faulty, i.e., do not satisfy $\varphi$. What should we do? Discard $\petriNet$, and try to re-project a new system
from scratch? In the next theorem (Theorem \ref{theorem:SemanticallySafestSubsystem}) we will show that we may still be able to 
save $\petriNet$ by automatically synthesizing the best $(b,r)$-bounded $p/t$-net $\petriNet'$ whose partial order behavior lies in between
$\partialorders_{sem}(\petriNet,c)\cap \partialorders(\varphi,c)$ and $\partialorders_{sem}(\petriNet,c)$. In other words,
the partial order behavior of $N'$ is a subset of the partial order behavior of $N$ which preserves all safe runs of $N$. Additionally, 
the partial order behavior of $N'$ has as few unsafe partial-order runs as possible. 
We call $N'$ the {\em semantically safest subsystem} of $N$. We notice that the net $N'$ does not need to be 
a sub-net of $N$, and indeed $N'$ can have even more places than $N$. Only the behavior of $N'$ is guaranteed to be 
a subset of the behavior of $N$. 
\begin{theorem}[Semantically Safest Subsystem]
\label{theorem:SemanticallySafestSubsystem}
Let $c,b,r\in \N$ and $sem\in \{ex,cau\}$. Given a $(b,r)$-bounded $p/t$-net $\petriNet=(P,T)$ and an MSO formula $\varphi$, 
we may automatically synthesize a $(b,r)$-bounded $p/t$-net $\petriNet'$ such that 
\begin{enumerate}[i)]
	\item \label{SemanticallySafestSubsystem-i} $N'$ is $c$-$sem$-minimal for $\partialorders(c,T,\varphi)\cap \partialorders_{sem}(N,c)$,
	\item \label{SemanticallySafestSubsystem-ii} $\partialorders_{sem}(N',c)\subseteq \partialorders_{sem}(N,c)$. 
\end{enumerate}
\end{theorem}
We consider that our notion of {\em semantically safest subsystem} is appropriate for three reasons. First, as mentioned above, 
\ref{theorem:SemanticallySafestSubsystem}.\ref{SemanticallySafestSubsystem-i} and  \ref{theorem:SemanticallySafestSubsystem}.\ref{SemanticallySafestSubsystem-ii}
imply that $\partialorders(c,T,\varphi) \cap \partialorders_{sem}(N,c) \subseteq \partialorders(N',c) \subseteq \partialorders(N,c)$. 
Second, the minimality condition says that if there is a $(b,r)$-bounded $p/t$-net $N'$ whose $c$-partial-order behavior precisely matches 
$\partialorders_{sem}(N,c) \cap \partialorders(\varphi,c)$ then such a $p/t$-net will be returned. In this case, our synthesis algorithm completely 
corrects the original $p/t$-net. Finally, but not less important, if all $c$-partially-ordered runs of $\petriNet$ indeed satisfy $\varphi$, 
then our synthesis algorithm returns a net $\petriNet'$ satisfying $\partialorders_{sem}(\petriNet',c)=\partialorders_{sem}(\petriNet,c)$. 
Thus the set of $c$-partial order behaviors of the synthesized net does not change if the original net is already correct 
(although the structure of the net per si may change). In Subsection \ref{subsection:BehavioralRepair} below we consider a 
related problem that finds analogies with the field of automatic program repair. 

\vspace{-10pt}
\subsection{Behavioral Repair}
\label{subsection:BehavioralRepair}

During the last decade a substantial amount of effort has been devoted to the development of methodologies 
for the automatic correction of bugs in computational systems \cite{CernyHenzingerRadhakrishnaRyzhykTarrach2013,GriesmayerBloemCook2006,JobstmannGriesmayerBloem2005,SamantaDeshmokhEmerson2008}.
Very recently, in the context of reactive systems, Jobstmann and von Essen have combined system synthesis and model checking to 
develop a methodology of program repair that preserves semantically correct runs \cite{EssenJobstmann2013}. Within their
methodology, given two LTL formulas $\varphi$ and $\psi$ and a reactive system $S$, one is asked to automatically synthesize an 
system $S'$ whose behavior is lower bounded by $\lang(\varphi)\cap \lang(S)$ and upper bounded by $\lang(\psi)$. Intuitively, 
while $\varphi$ specifies a set of correct behaviors that should be preserved whenever present in the original system, the formula $\psi$ specifies 
the set of behaviors that are allowed to be present in the repaired system. In Theorem \ref{theorem:BehavioralRepair} below we 
transpose the semantically preserving repair methodology devised in \cite{EssenJobstmann2013} to the realm of bounded $p/t$-nets 
with the partial order semantics.

\begin{theorem}[Behavioral Repair]
\label{theorem:BehavioralRepair}
Let $c,b,r\in \N$ and $sem\in \{ex,cau\}$. Given a $(b,r)$-bounded $p/t$-net $\petriNet=(P,T)$ and an MSO formula $\varphi$, 
we may automatically determine whether there exists a $(b,r)$-bounded $p/t$-net $\petriNet'$ such that 
\begin{enumerate}[i)]
	\item \label{BehavioralRepair-i} $N'$ is $c$-$sem$-minimal for $\partialorders(c,T,\varphi)\cap \partialorders_{sem}(N,c)$,
	\item \label{BehavioralRepair-ii} $\partialorders_{sem}(N',c)\subseteq \partialorders_{sem}(c,T,\psi)$. 
\end{enumerate}
In the case such a net exists, one may automatically construct it. 
\end{theorem}

While \ref{theorem:BehavioralRepair}.\ref{BehavioralRepair-i} and \ref{theorem:BehavioralRepair}.\ref{BehavioralRepair-ii} 
imply that $\partialorders(c,T,\varphi)\cap \partialorders_{sem}(N,c) \subseteq \partialorders(N',c) \subseteq \partialorders(c,T,\psi)$, 
the minimality condition in \ref{theorem:BehavioralRepair}.\ref{BehavioralRepair-i} implies that if $N'$ is successfully synthesized, then its 
behavior has as few partial-order runs contradicting $\varphi$ as possible.

\vspace{-10pt}
\subsection{Synthesis from Partial Order Contracts}
\label{subsection:AutomatedRepair}

Suppose that we are in the early stages of development of a concurrent system. We have arrived 
to the conclusion that every behavior satisfying a given MSO formula $\varphi^{\mathit{yes}}$ should be present 
in the system, but that no behavior in the system should satisfy a formula $\varphi^{\mathit{no}}$. 
Clearly we require that $\partialorders(\varphi^{yes}) \cap \partialorders(\varphi^{no})=\emptyset$. 
We say that the pair $(\varphi^{yes},\varphi^{no})$ is a partial order contract. 
We can try to develop a first prototype of our system by automatically synthesizing a $(b,r)$-bounded 
$p/t$-net $N$ containing all $c$-partial orders specified by $\varphi^{yes}$ but no partial order in $\varphi^{no}$.
The next theorem says that if such a net exists, then it can be automatically constructed. 

\begin{theorem}[Synthesis from Contracts]
\label{theorem:Contract}
Let $\varphi^{yes}$ and $\varphi^{no}$ be $\msopartialorders$ formulas with $\partialorders(c,T,\varphi^{yes})\cap \partialorders(c,T,\varphi^{no})=\emptyset$. 
Then one may automatically determine whether there exists a $(b,r)$-bounded 
$p/t$-net $\petriNet$ such that $\partialorders(c,T,\varphi^{yes}) \subseteq \partialorders_{sem}(\petriNet,c)$ and 
$\partialorders(c,T,\varphi^{no})\cap \partialorders_{sem}(\petriNet,c) = \emptyset$. In the case such a net exists one 
may construct it.  
\end{theorem}

\section{p/t-Nets and their Partial Order Semantics} 
\label{section:PetriNets}

Let $T$ be a finite set of transitions. Then a place over $T$ is a triple
$p=(p_0,\putsp,\takesp)$ where $p_0$ denotes the initial number of tokens in
$p$ and $\putsp,\takesp:T\rightarrow \N$ are functions which denote the number
of tokens that each transition $t\in T$ respectively puts in and takes from $p$. A
$p/t$-net over $T$ is a pair $N=(P,T)$ where $T$ is a set of transitions and
$P$ a finite multi-set of places over $T$. We assume through this paper that
for each transition $t\in T$, there exist places $p_1,p_2\in P$ for which
$\putsp_1(t) >  0$ and $\takesp_2(t) > 0$. A marking of $N$ is a function
$m:P\rightarrow \N$. A transition $t$ is enabled at marking $m$ if $m(p)\geq
\takesp(t)$ for each $p\in P$. The occurrence of an enabled transition at
marking $m$ gives rise to a new marking $m'$ defined as $m'(p)=m(p) -\takesp(t)
+ \putsp(t)$. The initial marking $m_0$ of $N$ is given by $m_0(p)=p_0$ for
each $p\in P$. A sequence of transitions $t_1t_1...t_{n}$ is an occurrence
sequence of $N$ if there exists a sequence of markings $m_0m_1...m_n$ such that
for each $i\in \{1,...,n\}$, $t_i$ is enabled at $m_{i-1}$ and if $m_{i}$ is obtained by the firing of $t_i$
at marking $m_{i-1}$. A marking $m$ is legal if it is obtained from $m_0$ by the 
firing of an occurrence sequence of $N$. A place $p$ of $N$ is $b$-bounded if $m(p)\leq b$
for each legal marking $m$ of $N$. A net $N$ is $b$-bounded if each of its places
is $b$-bounded. The union of two $p/t$-nets $N_1=(P_1,T)$ and $N_2=(P_2,T)$ having 
a common set of transitions $T$  is the $p/t$-net $N_1\cup N_2=(P_1 \cup P_2, T)$.
Observe that since we are dealing with the union of multisets, if a place $p$ occurs
with multiplicity $r_1$ in $P_1$ and with multiplicity $r_2$ in $P_2$ then the same 
place will occur with multiplicity $r_1+r_2$ in $P_1\cup P_2$. 

The notion {\em process}, upon which the partial order semantics of $p/t$-nets 
is derived, is defined in terms of objects called {\em occurrence nets}. An 
occurrence net is a DAG $O=(B\dot{\cup} V, F)$ where the vertex set $B\dot{\cup} V$
is partitioned into a set $B$, whose elements are called conditions, and 
a set $V$, whose elements are called events. The edge set $F \subseteq (B\times V)\cup (V\times B)$  is restricted in such 
a way that for every condition $b\in B$, 
$$|\{(b,v)\;|\;v\in V\}|\leq 1 \hspace{0.5cm} \mbox{ and } \hspace{0.5cm} |\{(v,b)\;|\; v\in V\}|\leq 1.$$
In other words, conditions in an occurrence net are unbranched. For a condition $b\in B$ we 
let $InDegree(b)$ denote the number of edges having $b$ as target. A process of a $p/t$-net 
is an occurrence net in which conditions are labeled with places of $N$ and events are
labeled with transitions of $N$ in such a way that the number of conditions labeled by a 
place $p\in N$ which immediately precede (follows) an event labeled by a transition $t$ is equal to $\takesp(t)$ ($\putsp(p)$). 
We define processes more precisely below.

\begin{definition}[Process \cite{GoltzReisig1983}] \label{definition:Process} A process of a $p/t$-net
$N=(P,T)$ is a labeled DAG $\process=(B\dot{\cup}V,F,\rho)$ where $(B\dot{\cup} V,F)$ is 
an occurrence net and $\rho:(B\cup V)\rightarrow (P\cup T)$ is a labeling function 
satisfying the following properties.
\begin{enumerate} 
	\item \label{process:itemThree} Places label conditions and transitions label events. 
		$$\rho(B)\subseteq P \hspace{1cm} \rho(V)\subseteq T$$ 
	\item \label{process:itemFour} For every $v\in V$, and every $p\in P$,
	$$|\{(b,v)\in F : \rho(b)\!=\!p\}|=\takesp(\rho(v))\hspace{0.2cm}\mbox{ and }\hspace{0.2cm} |\{(v,b)\in F : \rho(b)\!=\!p\}|=\putsp(\rho(v))$$
	\item For every $p\in P$, $$|\{b| \mathit{InDegree(b)}=0, \rho(b)=p \}|= p_0 .$$ 
 \end{enumerate}
\end{definition}

Let $R\subseteq X\times X$ be a binary relation on a set $X$. We 
denote by $\transitiveClosure(X)$ the transitive closure of $R$. If $\process=(B\cup V,F,\rho)$ 
is a process then the {\em causal order} of $\process$ is the partial order $\apartialorder_{\process}=(V,\transitiveClosure(F)|_{V\times V},\rho|_{V})$
which is obtained by taking the transitive closure of $F$ and subsequently by restricting $\transitiveClosure(F)$ to pairs of events of $V$.
In other words the causal order of a process $\process$ is the partial order induced by $\process$ on its events. We denote by 
$\partialorders_{\mathit{cau}}(\petriNet)$ the set of all partial orders derived from processes of $\petriNet$. We say that 
$\partialorders_{\mathit{cau}}(N)$ is the causal language of $N$.  

$$\partialorders_{\mathit{cau}}(N) = \{\apartialorder_{\process} | \process \mbox{ is a process of $N$}\}$$

Observe that several processes of a $p/t$-net $N$ may correspond to the same partial order. 
A sequentialization of a partial order $\apartialorder$ is any partial order $\apartialorder'=(V,<',l)$ for which $<\subseteq <'$. 
If $\petriNet$ is a $p/t$-net then an {\em execution} of $\petriNet$ is any sequentialization of 
 a causal order in $\partialorders_{\mathit{cau}}(N)$. We denote by $\partialorders_{\mathit{ex}}(N)$ the set 
of all executions of $N$. 

$$\partialorders_{\mathit{ex}}(N) = \{\apartialorder| \mbox{ $\apartialorder$ is a sequentialization of a 
causal order in $\partialorders_{\mathit{cau}}(N)$}\}.$$

We denote by $\partialorders_{\mathit{ex}}(N,c)$ the set of all $c$-partial orders in $\partialorders_{\mathit{ex}}(N)$
and by $\partialorders_{\mathit{cau}}(N,c)$ the set of all $c$-partial orders in $\partialorders_{\mathit{cau}}(N)$.  
We notice that when considering the execution semantics of a $b$-bounded $p/t$-net $N=(P,T)$, the 
set $\partialorders_{ex}(N,1)$ is simply the set of all occurrence sequences of $N$. 
Additionally, the inclusion $\partialorders_{ex}(N,c)\subseteq \partialorders_{ex}(N,c+1)$ may 
be proper for infinitely many values of $c$. In other words the execution behavior of a $p/t$-net may increase 
infinitely often with an increase in the parameter $c$. On the other hand, when considering the causal semantics 
of $N$, it can be shown  \cite{deOliveiraOliveira2010} that $$\partialorders_{cau}(N,b\cdot |P|)=\partialorders_{cau}(N,b\cdot |P|+i) = \partialorders_{cau}(N)$$ for any $i\in \N$.
Thus the causal behavior of a $b$-bounded $p/t$-net stabilizes for $c=b\cdot |P|$.

\section{Regular Slice Languages}
\label{section:RegularSliceLanguages}
A slice $\boldS=(V,E,l,s,t)$ is a DAG\footnote{A generalization of slices to arbitrary digraphs was considered in \cite{deOliveiraOliveira2013}, but
in this work we are only interested in slices that give rise to DAGs.} where $V=I\dot{\cup} C\dot{\cup} O$ is a set of 
vertices partitioned into an in-frontier $I$, a center $C$ and an out-frontier $O$; $E$ is a set of edges, $s,t:E\rightarrow V$ are functions 
that associate to each edge $e\in E$ a source vertex $e^s$ and a target vertex $e^t$, and $l:V\rightarrow \alphabettransitions \cup \N$
is a function that labels the center vertices in $C$ with elements of a finite set $T$, and the in- and out-frontier vertices with 
positive integers in such a way that $l(I)=\{1,...,|I|\}$ and $l(O)=\{1,...,|O|\}$. 
Additionally, we require that each frontier-vertex $v$ in $I\cup O$ is the endpoint of exactly one edge $e\in E$ and that 
no edge has both endpoints in the same frontier. Finally, in this work, we consider that the edges are directed from the in-frontier to the 
out frontier. In other words, for each edge $e\in E$, $e^s\in I\cup C$ and $e^{t}\in C\cup O$. From now on we will omit the 
source and target functions $s$ and $t$ from the specification of a slice and write simply $\boldS=(V,E,l)$. 

A slice $\boldS_1=(V_1,E_1,l_1)$ with frontiers $(I_1,O_1)$ can be glued to a slice $\boldS_2=(V_2,E_2,l_2)$ with frontiers $(I_2,O_2)$
provided $|O_1|=|I_2|$. In this case the glueing gives rise to the slice $\boldS_1\circ \boldS_2$ with 
frontiers $(I_1,O_2)$ which is obtained by fusing, for each $i\in \{1,...,|O_1|\}$, the unique edge $e_1\in E_1$ 
for which $l_1(e_1^t)=i$ with the unique edge $e_2\in E_2$ for which $l_2(e_2^s)=i$. 
Formally, the fusion of $e_1$ with $e_2$ proceeds as follows. First we create an edge $e_{12}$. Then we set 
$e_{12}^s=e_1^s$ and $e_{12}^{t}=e_2^t$. Finally we delete both $e_1$ and $e_2$. Thus in the glueing process 
the vertices in the glued frontiers disappear. 

A {\em unit slice} is a slice with exactly one vertex in its center. A slice is {\em initial} if
it has empty in-frontier and {\em final}  if it has empty out-frontier. The width of a slice $\boldS$ 
with frontiers $(I,O)$ is defined as $\width(\boldS) = \max\{|I|,|O|\}$. If $\labelAlphabet$ is a 
finite alphabet of symbols, then we let $\directedslicealphabet(c,\labelAlphabet)$ be the set of all 
unit slices of width at most $c$ whose unique center vertex is labeled with an element of $\labelAlphabet$. 
Observe that  $\directedslicealphabet(c,\labelAlphabet)$ is finite and has asymptotically $|\alphabettransitions|\cdot 2^{O(c\log c)}$
slices. We let $\directedslicealphabet(c,\alphabettransitions)^*$ denote the free monoid generated by $\directedslicealphabet(c,\alphabettransitions)$. We should 
{\it emphasize} that at this point the operation of the free monoid is simply the concatenation $\boldS\boldS'$ of 
slices and {\em should not be confused} with the composition $\boldS\circ \boldS'$. 
Thus the elements of $\directedslicealphabet(c,\alphabettransitions)^*$ are simply sequences $\boldS_1\boldS_2...\boldS_n$ of 
slices regarded as dumb letters. Additionally, the identity element of this monoid is simply the empty 
string $\lambda$, for which $\lambda\boldS=\boldS=\boldS\lambda$. 

We let $\lang(\directedslicealphabet(c,\alphabettransitions))$ be the set of all sequences of slices 
$\boldS_1\boldS_2...\boldS_n \in \directedslicealphabet(c,\alphabettransitions)^*$ 
for which $\boldS_i$ can be composed with $\boldS_{i+1}$ for $i\in \{1,...,n-1\}$, and for which $\boldS_1$ is initial and 
$\boldS_n$ is final. We call the elements of $\lang(\directedslicealphabet(c,\alphabettransitions))$ {\em unit decompositions}. 
If $\boldU = \boldS_1\boldS_2...\boldS_n$ is a unit decomposition in $\lang(\directedslicealphabet(c,\alphabettransitions))$ then we 
denote by $\composedU = \boldS_1\circ \boldS_2 \circ ... \circ \boldS_n$ the DAG obtained from $\boldU$ by composing 
all of its slices. A {\em slice language} over $\directedslicealphabet(c,\alphabettransitions)$ is any subset of 
$\lang(\directedslicealphabet(c,\alphabettransitions))$. The width $\width(\boldU)$ of a unit decomposition $\boldU=\boldS_1\boldS_2...\boldS_n$
is the maximum width of a slice occurring in $\boldU$: $\width(\boldU) = \max_{i} \width(\boldS_i)$. 
Each slice language $\lang$ represents a possibly infinite family of DAGs $\lang_{\graph}$ which is obtained by composing the slices in each unit decomposition in $\lang$. 

\begin{equation}
\label{equation:GraphLanguage}
\lang_{\graph} =  \{\composedU | \boldU \in \lang\}
\end{equation}

Additionally, $\lang$ also represents a possibly infinite family of partial orders $\lang_{\partialOrder}$
which is obtained by taking the transitive closure of each DAG in $\lang_{\graph}$. 

\begin{equation}
\label{equation:PartialOrderLanguage}
\lang_{\partialOrder} = \{\transitiveClosure(\composedU)\; |\;  \boldU \in \lang\}
\end{equation}

A slice language $\lang\subseteq \lang(\directedslicealphabet(c,\alphabettransitions))$ is regular if 
it can be defined by a finite automaton $\automaton$ over the slice alphabet $\directedslicealphabet(c,\alphabettransitions)$.

\begin{definition}[Slice Automaton]
\label{definition:SliceAutomaton}
Let $T$ be a finite set of symbols and let $c\in \N$. A slice automaton over a slice alphabet $\directedslicealphabet(c,T)$ is a finite automaton $\mathcal{A}=(Q,\Delta,q_0,F)$
where $Q$ is a set of states, $q_0\in Q$ is an initial state, $F\subseteq Q$ is a set of final states, and 
$\Delta\subseteq Q\times \directedslicealphabet(c,T) \times Q$ is a transition relation such that for 
every $q,q',q''\in Q$ and every $\boldS\in \directedslicealphabet(c,T)$:
\begin{enumerate}
	\item if $(q_0,\boldS,q)\in \Delta$ then $\boldS$ is an initial slice, 
	\item if $(q,\boldS,q')\in \Delta$ and $q'\in F$, then $\boldS$ is a final slice,
	\item if $(q,\boldS,q')\in \Delta$ and $(q',\boldS',q'')\in \Delta$, then $\boldS$ can be glued to $\boldS'$. 
\end{enumerate}
\end{definition} 

We denote by $\lang(\automaton)$ the slice language accepted by $\automaton$. We denote by 
$\lang_{\graph}(\automaton)$ and $\lang_{\partialOrder}(\automaton)$ respectively the set 
of DAGs derived from unit decompositions in $\lang(\automaton)$ and the set of partial orders 
obtained by taking the transitive closure of DAGs in $\lang_{\graph}(\automaton)$.

\section{Saturated and Transitively Reduced Slice Languages}
\label{section:SaturatedAndTransitiveReducedSliceLanguages}

Let $H$ be a DAG whose vertices are labeled with elements from a finite set $T$.
Then we let $\unitdecompositions(H,c)$ denote the set of all unit decompositions
$\boldU$ in $\lang(\directedslicealphabet(c,T))$ for which $\composedU = H$. 
The {\em set of all unit decompositions} of $H$ is defined as 

\begin{equation}
\label{equation:UnitDecompositions}
\unitdecompositions(H)=\bigcup_{c\geq 0} \unitdecompositions(c,T)
\end{equation}

We say that a slice language $\lang$ over $\directedslicealphabet(c,T)$ is {\em saturated} if for every DAG $H\in \lang_{\graph}$ 
we have that $\unitdecompositions(H)\subseteq \lang$. Notice that if a slice language $\lang$ over $\directedslicealphabet(c,T)$ 
is saturated, then for any $H\in \lang_{\graph}$ we have that $\unitdecompositions(H,c)=\unitdecompositions(H,c')$ for any $c'\geq c$. 
Let $H$ be a DAG. An ordering $\omega=(v_1,v_2,...,v_n)$ of the vertices of $H$ is a {\em topological ordering} if 
for any $i,j$ with $1\leq i<j\leq n$, there is no edge of $H$ whose source is $v_j$ and whose target is $v_i$. In other words,
in a topological ordering, the target of an edge has always a greater position in the ordering than its source. 
Notice that if $\boldU=\boldS_1\boldS_2...\boldS_n$ is a unit decomposition of a DAG $H$, and if $v_i$ is the center vertex of $\boldS_i$, 
then the ordering $\omega = (v_1,v_2,...,v_n)$ is always a topological ordering of $H$. We say that $\boldU$ is compatible with $\omega$. 
Conversely, given any topological ordering $\omega$ of $H$ there exists at least one unit decomposition $\boldU$ of $H$ that is compatible with 
$\omega$. We denote by $\unitdecompositions(H,\omega)$ the set of all unit decompositions of $H$ that are compatible with $\omega$. Notice 
that $\unitdecompositions(H) =  \bigcup_{\omega} \unitdecompositions(H,\omega)$ where $\omega$ ranges over all topological orderings 
of $H$. We say that a slice language is vertically saturated if for every $H\in \lang_{\graph}$ and every topological ordering $\omega$
of $H$, $\unitdecompositions(H,\omega)\cap \lang \neq \emptyset$ implies that $\unitdecompositions(H,\omega)\subseteq \lang$. Notice that a slice language 
may be vertically saturated without being saturated. 
In general, deriving from a slice automaton $\automaton$ a slice automaton $\automaton'$ such that 
$\lang(\automaton')$ is saturated and such that $\lang_{\partialOrder}(\automaton')= \lang_{\partialOrder}(\automaton)$  
is an uncomputable problem \cite{deOliveiraOliveira2010}. However it is always possible to derive from 
$\automaton$ a slice automaton $\automaton''$ such that $\lang(\automaton'')$
is vertically saturated and such that $\lang_{\partialOrder}(\automaton'') = \lang_{\partialOrder}(\automaton)$ \cite{deOliveiraOliveira2012}. 
We say that a slice automaton $\automaton$ is saturated if $\lang(\automaton)$ is saturated. We 
say that $\automaton$ is vertically saturated if $\lang(\automaton)$ is vertically saturated. 

The transitive reduction of a DAG $H=(V,E,l)$ is the minimal subgraph $\transitiveReduction(H)$ of $H$ with the same 
transitive closure as $H$. In other words $\transitiveClosure(\transitiveReduction(H))=\transitiveClosure(H)$. We say that 
a DAG $H$ is transitively reduced if $H=\transitiveReduction(H)$. Alternatively, we say that a transitively reduced DAG is a 
Hasse diagram. We say that a slice language $\lang$ is transitively reduced if 
every DAG in $\lang_{\graph}$ is transitively reduced. 
The transitive reduction of a slice language $\lang$ is the unique slice language 
$\transitiveReduction(\lang)$ which is transitively reduced, vertically saturated and such that for each DAG $H\in \lang_{\graph}$ 
and each topological ordering $\omega$ of $H$, $\unitdecompositions(H,\omega)\cap \lang\neq \emptyset$ implies that $\unitdecompositions(\transitiveReduction(H),\omega)\cap \lang \neq \emptyset$. 
Notice that by our definition of transitive reduction, if $\lang$ is saturated, then $\transitiveReduction(\lang)$ is also saturated.
We say that a slice automaton $\automaton$ is transitively reduced if $\lang(\automaton)$ is transitively reduced.  

\begin{lemma}[Transitive Reduction of Slice Languages \cite{deOliveiraOliveira2012}]
\label{lemma:TransitiveReductionSliceLanguages}
Let $\lang$ be a regular slice language represented by a finite automaton $\automaton$ over $\directedslicealphabet(c,T)$. Then 
there exists a finite automaton $\transitiveReduction(\automaton)$ on $2^{O(c\log c)}\cdot|\automaton|$ states with 
$\lang(\transitiveReduction(\automaton)) = \transitiveReduction(\lang)$. 
\end{lemma}

Let $T$ be a finite set of transitions. We denote by $\partialorders(c,T)$ the set of all $c$-partial orders whose vertices are labeled with elements from $T$.

\begin{lemma}[\cite{deOliveiraOliveira2012}]
\label{lemma:InitialLanguage}
For any finite set $T$ and any $c\in \N$, one can construct  
a saturated transitively reduced slice automaton $\automaton(c,T)$ 
over $\directedslicealphabet(c,T)$ such that $\lang_{\partialOrder}(\automaton(c,T)) = \partialorders(c,T)$. 
\end{lemma}

\begin{definition}[$c$-Complementation]
\label{definition:cComplement}
Let $\partialorders \subseteq \partialorders(c,T)$. Then we let $\overline{\partialorders}^c=\partialorders(c,T)\backslash \partialorders$ 
be the $c$-complement of $\partialorders$. 
\end{definition}

The following lemma says that operations performed on transitively reduced saturated slice languages are reflected on the partial order languages
they represent. Below $\automaton\cup \automaton'$, $\automaton\cap \automaton'$ and $\automaton\backslash \automaton'$ denote automata whose {\em slice language} 
(i.e. the syntactic language) is equal to $\lang(\automaton)\cup \lang(\automaton')$, $\lang(\automaton)\cap \lang(\automaton')$ and 
$\lang(\automaton)\backslash \lang(\automaton')$ respectively. 

\begin{lemma}[Properties of Saturated Slice Languages \cite{deOliveiraOliveira2012}]
\label{lemma:PropertiesSaturatedSliceLanguages}
Let $\automaton$ and $\automaton'$ be two transitively-reduced slice automata over $\directedslicealphabet(c,\alphabettransitions)$. 
Assume that $\automaton$ is saturated. 
\begin{enumerate}
\itemsep0.2em
\item $\lang_{\partialOrder}(\automaton\cup \automaton') = \lang_{\partialOrder}(\automaton) \cup \lang_{\partialOrder}(\automaton')$
\item $\lang_{\partialOrder}(\automaton\cap \automaton') = \lang_{\partialOrder}(\automaton) \cap \lang_{\partialOrder}(\automaton')$
\item $\lang_{\partialOrder}(\automaton(c,T)\backslash \automaton) = \overline{\lang_{\partialOrder}(\automaton)}^c$. 
\item $\lang_{\partialOrder}(\automaton) \subseteq \lang_{\partialOrder}(\automaton')$ if and only if 
	$\lang(\automaton) \subseteq \lang(\automaton')$. 
\item $\lang_{\partialOrder}(\automaton) \cap \lang_{\partialOrder}(\automaton')=\emptyset$ if and only if 
	$\lang(\automaton) \cap \lang(\automaton')=\emptyset$. 
\item If $\automaton'$ is saturated then $\automaton\cup \automaton'$ and $\automaton \cap \automaton'$ are also saturated.
\end{enumerate}
\end{lemma}

Lemma \ref{lemma:PropertiesSaturatedSliceLanguages} implies that union, intersection and $c$-complementation of partial order 
languages represented by transitively reduced saturated slice automata are computable, and inclusion and emptiness of 
intersection of these partial order languages are decidable. Theorem \ref{theorem:RegularSliceLanguagesPetriNets} establishes 
a close correspondence between the partial order behavior of bounded $p/t$-nets and regular slice languages. 
\newpage

\begin{theorem}[$p/t$-nets and Regular Slice Languages
\cite{deOliveiraOliveira2010,deOliveiraOliveira2012}]
\label{theorem:RegularSliceLanguagesPetriNets} 
\begin{enumerate}[i)] 
	\item {\bf Expressibility: }
		\label{theorem:Expressibility} Let $\petriNet=(P,T)$ be a $b$-bounded $p/t$-net and 
		$\mathit{sem}\in \{\mathit{ex},\mathit{cau}\}$. Then one can construct a saturated 
		transitively reduced slice automaton $\automaton_{\mathit{sem}}(N,c)$ over $\directedslicealphabet(c,T)$ such that 
		$\lang_{\partialOrder}(\automaton_{\mathit{sem}}(N,c))=\partialorders_{\mathit{sem}}(N,c)$.\\ 
	\item {\bf Verification: }\label{theorem:Verification} Let $\petriNet=(P,T)$
		be a $b$-bounded $p/t$-net, $\automaton$ be a slice automaton over $\directedslicealphabet(c,T)$, and 
		$\mathit{sem}\in \{\mathit{ex},\mathit{cau}\}$.
		\begin{enumerate}
			\item It is decidable whether $\partialorders_{sem}(N,c) \cap \lang_{\partialOrder}(\automaton) = \emptyset$,
			\item It is decidable whether $\lang_{\partialOrder}(\automaton) \subseteq \partialorders_{sem}(N,c)$,
			\item If $\automaton$ is saturated then it is decidable whether $\partialorders_{sem}(N,c)\subseteq \lang_{\partialOrder}(\automaton)$.\\
		\end{enumerate}
	\item {\bf Synthesis: }\label{theorem:Synthesis} Let
		$\automaton$ be a slice automaton, $c,b,r\in \N$ and $\mathit{sem}\in \{\mathit{ex},\mathit{cau}\}$. 
		Then one may automatically determine whether there exists a $(b,r)$-bounded $p/t$-net $\petriNet$ 
		which is $c$-$\mathit{sem}$-minimal for $\lang_{\partialOrder}(\automaton)$. In the case such a net exists, one may automatically construct it. 
\end{enumerate}
\end{theorem}

Observe that the synthesis result stated in Theorem \ref{theorem:RegularSliceLanguagesPetriNets}.\ref{theorem:Synthesis}
can be understood as the inverse of the expressibility result stated in Theorem \ref{theorem:RegularSliceLanguagesPetriNets}.\ref{theorem:Expressibility}. 
On the one hand, Theorem \ref{theorem:RegularSliceLanguagesPetriNets}.\ref{theorem:Expressibility} can be used to construct a 
slice automaton $\automaton$ representing the $c$-partial order behavior of $N = (P,T)$. On the other hand, Theorem \ref{theorem:RegularSliceLanguagesPetriNets}.\ref{theorem:Synthesis} can 
be used to recover from $\automaton$ a $p/t$-net $N'$ whose $c$-partial order behavior is equal to the $c$-partial-order behavior of $N$.

\vspace{-10pt}
\section{Monadic Second Order Logic of Graphs}
\label{section:MSO}

The monadic second order logic of graphs $\msopartialorders$ extends first order logic by allowing 
quantification over sets of vertices. The logic $\msographs$ is an extension 
of $\msopartialorders$ that also allows quantification over sets of edges. We refer to \cite{CourcelleEngelfriet2012}
for an extensive treatment of these logics. 
In this section we 
will use $\msographs$ to describe properties of DAGs, while we will use $\msopartialorders$ to 
describe properties of partial orders.

We will represent a partial order $\apartialorder$  by a relational structure $\apartialorder=(V,<,l)$
where $V$ is a set of vertices, $<\subset V\times V$ is an ordering relation and 
$l\subseteq V\times \alphabettransitions$  is a vertex labeling relation where $T$ is a finite set of symbols (which should be regarded as the 
labels of transitions in a concurrent system). First order variables representing individual vertices will be taken from the set $\{x_1,x_2,...\}$ while 
second order variables representing sets of vertices will be taken from the set $\{X_1,X_2,...\}$. 
The set of $\msopartialorders$ formulas is the smallest set of formulas containing: 
\begin{itemize}
	\item the atomic formulas $x_i\in X$, $x_i < x_j$, $l(x_i,a)$ for each $i,j\in \N$ with $i\neq j$ and each $a\in \alphabettransitions$, 
	\item the formulas $\varphi \wedge \psi$, $\varphi \vee \psi$, $\neg \varphi$, $\exists x_i.\varphi(x_i)$ and $\exists X_i.\varphi(X_i)$,
	  	where $\varphi$ and $\psi$ are $\msopartialorders$ formulas. 
\end{itemize}

An $\msopartialorders$ sentence is a $\msopartialorders$ formula $\varphi$ without free variables. If $\varphi$ is a sentence, and 
$\apartialorder=(V,<,l)$ a partial order, then we denote by $\apartialorder \models \varphi$ the fact that $\apartialorder$ satisfies $\varphi$. 

We will represent a general DAG $G$ by a relational structure $G=(V,E,s,t,l)$
where $V$ is a set of vertices, $E$ a set of edges, $s,t \subseteq E\times V$ are respectively the source 
and target relations, $l\subseteq V\times \alphabettransitions$  is a vertex labeling relation, where $\alphabettransitions$ is a
finite set of symbols. 
If $e$ is an edge in $E$ and $v$ is a vertex in $V$ then $s(e,v)$ is true if $v$ is the source of $e$ and 
$t(e,v)$ is true if $v$ is the target of $e$. If $v\in V$ and $a\in \alphabettransitions$ then $l(v,a)$ is true if $v$ is labeled with $a$.
First order variables representing individual vertices will be taken from the set $\{x_1,x_2,...\}$ and first order variables 
representing edges, from the set $\{y_1,y_2,...\}$. Second order variables representing sets of vertices will 
be taken from the set $\{X_1,X_2,...\}$ and second order variables representing sets of edges, from the 
set $\{Y_1,Y_2,...\}$. 
The set of $MSO_2$ formulas is the smallest set of formulas containing: 
\begin{itemize}
	\item the atomic formulas $x_i\in X_j$, $y_i\in Y_j$, $s(y_i,x_j)$, $t(y_i,x_j)$, $l(x_i,a)$ for each $i,j\in \N$ and $a\in \alphabettransitions$, 
	\item the formulas $\varphi \wedge \psi$, $\varphi \vee \psi$, $\neg \varphi$, $\exists x_i.\varphi(x_i)$ and $\exists X_i.\varphi(X_i)$,
	  $\exists y_i.\varphi(Y_i)$ and $\exists Y_i.\varphi(Y_i)$, where $\varphi$ and $\psi$ are $\msographs$ formulas. 
\end{itemize}

An $\msographs$ sentence is a formula $\varphi$ without free variables. If $\varphi$ is a sentence, then 
we denote by $G\models \varphi$ the fact that $G$ satisfies $\varphi$.

\section{MSO Logic and Slice Languages}

Lemma \ref{lemma:MSO-Regular} below, which was proved in a more general context \cite{deOliveiraOliveira2013}, 
states that the set of all unit decompositions $\boldU$ whose graph $\composedU$ satisfy a given $\msographs$ 
formula $\varphi$ is a regular slice language.

\begin{lemma}[\cite{deOliveiraOliveira2013}]
\label{lemma:MSO-Regular}
Given a $\msographs$ formula $\varphi$, one can effectively construct a slice automaton $\automaton(c,\alphabettransitions,\varphi)$ 
over $\directedslicealphabet(c,\alphabettransitions)$ such that 
$$\lang(\automaton(c, \alphabettransitions,\varphi)) = \{\boldU\in \lang(\directedslicealphabet(c,\alphabettransitions))\; | \; \composedU \models \varphi\}.$$
\end{lemma}

We say that a DAG $H=(V,E)$ can be covered by $c$ paths if there exist simple paths $\path_1,...,\path_c$ in $H$ 
with $\path_i=(V_i,E_i)$ such that $V=\cup_i V_i$ and $E=\cup_i E_i$. Proposition \ref{proposition:UnionOfKPathsSaturation}
below establishes a correspondence between $c$-coverable DAGs and their sets of unit decompositions. 

\begin{proposition}
\label{proposition:UnionOfKPathsSaturation}
Let $H$ be a DAG. If $H$ can be covered by $c$ paths, then any unit decomposition of $H$ has width at most $c$. 
\end{proposition}

We let $\gamma(c)$ be the $\msographs$ sentence
which is true on a DAG $H$ whenever $H$ can be covered by $c$ paths. 
Then we have that $\lang(\automaton(c,\alphabettransitions,\varphi \wedge \gamma(c)))$ is the set of all unit 
decompositions in $\lang(\automaton(c,\alphabettransitions,\varphi))$ whose corresponding DAG can be covered by $c$-paths.

\begin{lemma}
\label{lemma:MSORegularPaths}
For any $\msographs$ formula $\varphi$ and any positive integer $c\in \N$, the slice automaton $\automaton(c,\alphabettransitions,\varphi\wedge\gamma(c))$ is 
saturated. 
\end{lemma}

Recall that if $H$ is a DAG, then $\transitiveReduction(H)$ denotes the transitive reduction of $H$. 

\begin{proposition}[Partial Orders vs Hasse Diagrams]
\label{proposition:PartialOrdersVsHasseDiagrams}
For any $\msopartialorders$ formula $\varphi$ expressing a partial order property, there is 
an $\msographs$ formula $\varphi^{\mathit{gr}}$ expressing a property of DAGs such that for 
any partial order $\apartialorder\in \partialorders(c,T)$, $\apartialorder\models\varphi$ 
if and only if $\transitiveReduction(\apartialorder) \models \varphi^{\mathit{gr}}$. 
\end{proposition}

Let $c\in \N$, $T$ be a finite set, and $\varphi$ be $\msopartialorders$ formula. 
We denote by $\partialorders(c,T,\varphi)$ the set of all $c$-partial orders satisfying $\varphi$ whose
vertices are labeled with elements from $T$. We denote by $\rho$ be the $\msographs$ formula which is 
true on a DAG $H$ whenever $H$ is transitively reduced, i.e., whenever $H=\transitiveReduction(H)$.

\begin{lemma}
\label{lemma:SaturatedAutomatonPartialOrders}
Let $\varphi$ be a $\msopartialorders$ formula expressing a partial order property, and $\varphi^{\mathit{gr}}$ 
be the $\msographs$ formula of Proposition \ref{proposition:PartialOrdersVsHasseDiagrams}. Then 
$\automaton(c,T,\varphi^{\mathit{gr}}\wedge \rho \wedge \gamma(c))$ is a saturated transitively reduced 
slice automaton and $\partialorders(c,T,\varphi) = \lang_{\partialOrder}(\automaton(c,T,\varphi^{\mathit{gr}}\wedge \rho \wedge \gamma(c))) $. 
\end{lemma}

\begin{lemma}[Verifying Regular Slice Languages]
\label{lemma:ModelCheckingRegularSliceLanguages}
Let $\varphi$ be a $\msopartialorders$ formula, and let $\automaton$ be a transitively reduced saturated slice automaton over 
$\directedslicealphabet(c,\alphabettransitions)$. 
\begin{enumerate}[i)]
\item One may effectively verify whether $\lang_{po}(\automaton)\cap \partialorders(c,T,\varphi) = \emptyset$.
\item One may effectively verify whether $\lang_{po}(\automaton) \subseteq \partialorders(c,T,\varphi)$.
\item One may effectively verify whether $\partialorders(c,T,\varphi) \subseteq \lang_{po}(\automaton)$.
\end{enumerate} 
\end{lemma}

\section{Proofs of our Main Results}
\label{section:MainResults}

Finally, we are in a position to prove our main results. First we will state a lemma that we call the 
separation lemma.

\begin{lemma}[Separation Lemma]
\label{lemma:SeparationLemma}
Let $\automaton$ and $\automaton'$ be two slice automata over $\directedslicealphabet(c,T)$. And suppose that 
$\automaton'$ is saturated. Then one can decide whether there exists a $(b,r)$-bounded $p/t$-net $N$ such that
\begin{enumerate}[i)]
	\item $N$ is $c$-sem-minimal for $\lang_{\partialOrder}(\automaton)$,
	\item $\partialorders(N,c) \cap \lang_{\partialOrder}(\automaton') = \emptyset$.
\end{enumerate}
In the case such a net $N$ exists one can automatically construct it.
\end{lemma}
\begin{proof}
First we apply Theorem \ref{theorem:RegularSliceLanguagesPetriNets}.\ref{theorem:Synthesis} to determine if there exists 
a $(b,r)$-bounded $p/t$-net $N$ that is $c$-sem-minimal for $\lang_{\partialOrder}(\automaton)$. In case such a net 
exists we construct it. Now using Theorem \ref{theorem:RegularSliceLanguagesPetriNets}.\ref{theorem:Expressibility} we 
construct a slice automaton $\automaton''$ show causal/execution behavior is precisely that of $N$. Finally, since 
$\automaton'$ is transitively reduced and saturated we can use Lemma \ref{lemma:PropertiesSaturatedSliceLanguages} to test 
whether $\lang_{\partialOrder}(\automaton')\cap \lang_{\partialOrder}(\automaton'')$. Notice that by the minimality of
$\lang_{\partialOrder}(N)$ if this intersection is not empty, then the problem has no solution. $\square$
\end{proof}

\paragraph{\textbf{Proof of Theorem \ref{theorem:MSOAndPetriNets}}}
Let $\petriNet=(P,T)$ be a $b$-bounded $p/t$-net and let $\mathit{sem}\in \{\mathit{ex},\mathit{cau}\}$. By Theorem \ref{theorem:RegularSliceLanguagesPetriNets} 
we can construct a saturated, transitively reduced slice automaton $\automaton_{\mathit{sem}}(\petriNet,c)$ such that 
$\lang_{\partialOrder}(\automaton_{\mathit{sem}}(\petriNet,c)) = \partialorders_{\mathit{sem}}(N,c)$.
Now by Lemma \ref{lemma:ModelCheckingRegularSliceLanguages} we can effectively determine  
whether $\lang_{\partialOrder}(\automaton_{\mathit{sem}}(\petriNet,c))\cap \partialorders(c,T,\varphi) = \emptyset$,
\\whether $\lang_{\partialOrder}(\automaton_{\mathit{sem}}(\petriNet,c)) \subseteq \partialorders(c,T,\varphi)$ 
or whether $\partialorders(c,T,\varphi) \subseteq \lang_{\partialOrder}(\automaton_{\mathit{sem}}(\petriNet,c))$.
$\square$

\paragraph{\textbf{Proof of Theorem \ref{theorem:MSOSynthesis}}}

Let $\varphi$ be a $\msopartialorders$ formula. By Lemma \ref{lemma:SaturatedAutomatonPartialOrders} 
one can construct a saturated, transitively reduced slice automaton $\automaton = \automaton(c,\alphabettransitions,\varphi\wedge \rho \wedge \gamma(c))$ 
over $\directedslicealphabet(c,\alphabettransitions)$ such that 
$\lang_{\partialOrder}(\automaton) = \partialorders(c,T,\varphi)$. 
By Theorem \ref{theorem:RegularSliceLanguagesPetriNets} one may automatically determine whether there exists a $(b,r)$-bounded 
$p/t$-net $\petriNet$ which is $c$-$\mathit{sem}$-minimal for 
$\lang_{\partialOrder}(\automaton)$, and in the case such a net exists, one may automatically construct it.
$\square$

\paragraph{\textbf{Proof of Theorem \ref{theorem:SemanticallySafestSubsystem}: }}
By Lemma \ref{lemma:SaturatedAutomatonPartialOrders} one can construct a saturated, transitively reduced slice automaton 
$\automaton = \automaton(c,\alphabettransitions,\varphi^{\mathit{gr}}\wedge \rho \wedge \gamma(c))$ such that
$\lang_{\partialOrder}(\automaton) = \partialorders(c,T,\varphi)$.
By Theorem \ref{theorem:RegularSliceLanguagesPetriNets}.\ref{theorem:Expressibility}, one 
can construct a saturated, transitively reduced slice automaton $\automaton'$
such that $\lang_{\partialOrder}(\automaton') = \partialorders_{sem}(N,c)$. 
Since both $\automaton$ and $\automaton'$ are
saturated and transitively reduced, by Lemma \ref{lemma:PropertiesSaturatedSliceLanguages}, we have that the 
slice automaton $\automaton\cap \automaton'$ is saturated and transitively reduced. Additionally 
$\lang_{\partialOrder}(\automaton\cap \automaton') = \lang_{\partialOrder}(\automaton)\cap \lang_{\partialOrder}(\automaton') = 
\partialorders(c,T,\varphi)\cap \partialorders_{\mathit{sem}}(\petriNet,c)$.
Additionally, by Lemma \ref{lemma:PropertiesSaturatedSliceLanguages} we can construct a transitively reduced and saturated slice automaton 
$\overline{\automaton'}^c$ such that $\lang(\overline{\automaton}^c) = \partialorders(c,T) \backslash \partialOrder(N,c)$. 
Thus as a last step we may apply Lemma \ref{lemma:SeparationLemma} to determine whether 
there exists a $(b,r)$-bounded $p/t$-net $N$ that is $c$-sem-minimal for  $\lang_{\partialOrder}(\automaton\cap \automaton')$ 
and such that $\partialorders(N,c)\cap \lang_{\partialOrder}(\overline{\automaton'}^c) = \emptyset$, and in the 
case that such a net exists, we can effectively construct it.  
$\square$

\paragraph{\textbf{Proof of Theorem \ref{theorem:BehavioralRepair}: }}
By Lemma \ref{lemma:SaturatedAutomatonPartialOrders} one can construct a saturated, transitively reduced slice automata 
$\automaton_{\varphi} = \automaton(c,\alphabettransitions,\varphi^{\mathit{gr}}\wedge \rho \wedge \gamma(c))$
and $\automaton_{\psi} = \automaton(c,\alphabettransitions,\psi^{\mathit{gr}}\wedge \rho \wedge \gamma(c))$  such that
$\lang_{\partialOrder}(\automaton_{\varphi}) = \partialorders(c,T,\varphi)$ and $\lang_{\partialOrder}(\automaton_{\psi}) = \partialorders(c,T,\psi)$ 
respectively. By Theorem \ref{theorem:RegularSliceLanguagesPetriNets}.\ref{theorem:Expressibility}, one 
can construct a saturated, transitively reduced slice automaton $\automaton'$
such that $\lang_{\partialOrder}(\automaton') = \partialorders_{sem}(N,c)$. 
Since both $\automaton$ and $\automaton'$ are
saturated and transitively reduced, by Lemma \ref{lemma:PropertiesSaturatedSliceLanguages}, we have that the 
slice automaton $\automaton_{\varphi}\cap \automaton'$ is saturated and transitively reduced. Additionally 
$\lang_{\partialOrder}(\automaton_{\varphi}\cap \automaton') = \lang_{\partialOrder}(\automaton_{\varphi})\cap \lang_{\partialOrder}(\automaton') = 
\partialorders(c,T,\varphi)\cap \partialorders_{\mathit{sem}}(\petriNet,c)$.
Thus as a last step we may apply Lemma \ref{lemma:SeparationLemma} to determine whether 
there exists a $(b,r)$-bounded $p/t$-net $N$ that is $c$-sem-minimal for  $\lang_{\partialOrder}(\automaton_{\varphi}\cap \automaton')$ 
and such that $\partialorders(N,c)\cap \lang_{\partialOrder}(\automaton_{\psi}) = \emptyset$, and in the 
case that such a net exists, we can effectively construct it.  
$\square$

\paragraph{\textbf{Proof of Theorem \ref{theorem:Contract}: }}
Let $\varphi^{yes}$ and $\varphi^{no}$ be two $\msopartialorders$ formulas specifying respectively a set of 
good partial order behaviors and a set of bad partial order behaviors. 
By Lemma \ref{lemma:SaturatedAutomatonPartialOrders} we can construct saturated, transitively reduced 
slice automata 
$$\automaton^{yes} = \automaton(c,T,[\varphi^{\mathit{yes}}]^{\mathit{gr}}\wedge \rho \wedge \gamma(c))\mbox{\hspace{0.5cm} and\hspace{0.5cm} }
\automaton^{no} = \automaton(c,T,[\varphi^{\mathit{no}}]^{\mathit{gr}}\wedge \rho \wedge \gamma(c))$$ such that 
$\lang_{\partialOrder}(\automaton^{yes})=\partialorders(c,T,\varphi^{yes})$ and $\lang_{\partialOrder}(\automaton^{\mathit{no}}, \partialorders(c,T,\varphi^{no}))$. 
Now by Theorem \ref{theorem:RegularSliceLanguagesPetriNets}.\ref{theorem:Synthesis} we can synthesize a $(b,r)$-bounded 
$p/t$-net $N$ that is $c$-$\mathit{sem}$-minimal with respect to $\lang_{\partialOrder}(\automaton^{\mathit{yes}})$. 
Since $\automaton^{\mathit{no}}$ is saturated we can apply Lemma \ref{lemma:SeparationLemma} to determine whether 
there exists a $(b,r)$-bounded $p/t$-net $N$ that is $c$-sem-minimal for $\lang_{\partialOrder}(\automaton^{\mathit{yes}})$ 
and such that $\partialorders(N,c)\cap \lang_{\partialorders}(\automaton^{\mathit{no}}) = \emptyset$. In the case 
such a net exists we can effectively construct it.  
$\square$

\vspace{-15pt}
\section{Conclusion}
\label{section:Conclusion}

In this work we have shown that both model checking of the $c$-partial-order behavior of bounded
$p/t$-nets and the synthesis of bounded $p/t$-nets from MSO definable sets of $c$-partial-orders are computationally feasible.
By combining these two results, we introduced the {\em semantically safest subsystem} problem as a new primitive for the 
study of automated correction of computational systems. Additionally we were able to lift the theory of automatic program repair
developed in \cite{EssenJobstmann2013} to the realm of bounded $p/t$-nets and to develop a methodology of synthesis by contracts 
that is suitable for the partial order theory of concurrency.

\bibliographystyle{abbrv} 
\bibliography{VerificationSynthesisCorrectionPetriNetsMSO}

\appendix

\section{Proofs of Auxiliary Results}

\paragraph{\textbf{Proof of Proposition \ref{proposition:UnionOfKPathsSaturation}: }}
Let $H$ be the union of $c$ paths $\path_1,...,\path_c$ and let $\boldU=\boldS_1\boldS_2...\boldS_n$ be a unit decomposition
of $H$. Let $v_i$ be the center vertex of $\boldS_i$. Then the ordering $\omega=(v_1,...,v_n)$ is a topological ordering of $H$. This implies 
that for any $i\in \{1,...,n\}$, and any $j\in \{1,...,c\}$ there exists at most one edge from $p_j$ whose source is in $\{v_1,...,v_i\}$ 
and whose target is in $\{v_{i+1},...,v_n\}$. Thus there exists at most $c$ edges in $p_1\cup ...\cup p_c$ with whose source is in $\{v_1,...,v_i\}$
and whose target is in $\{v_{i+1},...,v_n\}$. This implies that $|E(\{v_1,...,v_i\},\{v_{i+1},...,v_n\})|\leq c$ for each $i\in \{1,...,n\}$. 
and thus $\omega$ has cut-width at most $c$ with respect to $H$ since the  $w(\boldU)$ is equal to the cut-width of $\omega$ we 
have that that $w(\boldU)\leq c$. 
$\square$

\paragraph{\textbf{Proof of Lemma \ref{lemma:MSORegularPaths}: }}
By Lemma \ref{lemma:MSO-Regular},  $\lang(c,\alphabettransitions,\varphi\wedge\gamma(c))$ is a regular slice language over $\directedslicealphabet(c,T)$.   
Thus we just need to show that  $\lang(c,\alphabettransitions,\varphi\wedge\gamma(c))$ is saturated. 
A unit decomposition $\boldU$ belongs to $\lang(c,\alphabettransitions,\varphi\wedge \gamma(c))$ 
if and only if $\boldU$ satisfies the following three properties:  
$\boldU\in \lang(c,\alphabettransitions)$, $\composedU$ can be covered by $c$ paths and $\composedU\models \varphi$.
Since $\composedU$ can be covered by $c$ paths, it follows from Proposition 
\ref{proposition:UnionOfKPathsSaturation} that  $\unitdecompositions(\composedU)\subseteq \lang(\directedslicealphabet(c,\alphabettransitions))$.
Now let $\boldU'$ be an arbitrary unit decomposition in $\unitdecompositions(\composedU)$. Since
$\composedUprime = \composedU$, we have that $\composedUprime$ is the union of $c$ paths and satisfies $\varphi$. 
Thus $\composedU'\in \lang(c,\alphabettransitions,\varphi \wedge \gamma(c))$. Since $\boldU'$ was taken to 
be an arbitrary unit decomposition in $\unitdecompositions(\composedU)$, we have that $\lang(c,\alphabettransitions,\varphi\wedge \gamma(c))$
is saturated. $\square$

\paragraph{\textbf{Proof of Proposition \ref{proposition:PartialOrdersVsHasseDiagrams}: }}
Let $\apartialorder=(V,<,l)$ be a partial order and $\transitiveReduction(\apartialorder)=(V,E,l)$ be the transitive reduction of $\apartialorder$. 
Then for any two vertices $v,v'\in V$, we have that $v < v'$ if and only if there is a path $v=v_1e_1v_2...e_{n-1}v_n=v'$  from $v$ to 
$v'$ in $\transitiveReduction(\apartialorder)$. Now let $path(x_1,X,Y,x_2)$ be a $\msographs$ formula which is true in a DAG $H$ whenever there is a path 
starting at $x_1$, finishing at $x_2$, with internal vertices $X$ and internal edges $Y$. For a $\msopartialorders$ formula $\varphi$, let $\varphi^{\mathit{gr}}$ be the $\msographs$ formula 
which is obtained from $\varphi$ by replacing each occurrence of the atomic formula $x_1<x_2$ in $\varphi$ by the formula 
$\exists X\, \exists Y\, path(x_1,X,Y,x_2)$. Now have that $\apartialorder \models \varphi $ if and only if 
$\transitiveReduction(\apartialorder)\models \varphi^{\mathit{gr}}$. $\square$

\paragraph{\textbf{Proof of Lemma \ref{lemma:SaturatedAutomatonPartialOrders}: }}
Let $\rho$ be the $\msographs$ formula which 
is true in a DAG $H$ whenever $H$ is transitively reduced. Let $\varphi^{gr}$
be the formula obtained from $\varphi$ as in Proposition \ref{proposition:PartialOrdersVsHasseDiagrams}. 
Then we have that a DAG $H$ satisfies $\gamma(c)\wedge \rho \wedge \varphi^{\mathit{gr}}$
if and only if $H$ can be covered by $c$ paths, $H$ is transitively reduced and if 
the partial order $\transitiveClosure(H)$ induced by $H$ satisfies $\varphi$.
By Lemma \ref{lemma:MSORegularPaths}, the slice language $\lang(c,\alphabettransitions,\varphi\wedge \rho\wedge \gamma(c))$
is saturated, regular, and consists precisely of the unit decompositions yielding 
a graph satisfying $\varphi\wedge \rho\wedge \psi(c)$. Thus we just need to set $\automaton(c,T,\varphi)$ as the 
minimal deterministic finite automaton generating $\lang(c,\alphabettransitions,\varphi\wedge \rho\wedge \gamma(c))$. $\square$

\paragraph{\textbf{Proof of Lemma \ref{lemma:ModelCheckingRegularSliceLanguages}: }}
Let $\automaton'=\automaton(c,\alphabettransitions,\varphi^{\mathit{gr}}\wedge \rho \wedge \gamma(c))$. 
By Lemma \ref{lemma:SaturatedAutomatonPartialOrders} $\automaton'$ is saturated, transitively reduced 
and $\lang(\automaton') = \partialorders(c,T,\varphi)$. Since $\automaton$ is also transitively reduced 
and saturated, by Lemma \ref{lemma:PropertiesSaturatedSliceLanguages}, 
$\lang_{\partialOrder}(\automaton)\cap \lang_{\partialOrder}(\automaton') = \emptyset$ if and only if 
$\lang(\automaton) \cap\lang(\automaton')= \emptyset$, 
$\lang_{\partialOrder}(\automaton)\subseteq \lang_{\partialOrder}(\automaton')$ if and only if 
$\lang(\automaton) \subseteq \lang(\automaton')$ and $\lang_{\partialOrder}(\automaton')\subseteq \lang_{\partialOrder}(\automaton)$ 
if and only if $\lang(\automaton) \subseteq \lang(\automaton')$. Thus we have reduced 
emptiness of intersection and inclusion of the partial order languages represented by $\automaton$  and $\automaton'$
to the emptiness of intersection and inclusion of the regular slice languages accepted by $\automaton$ and $\automaton'$. 
$\square$

\end{document}